\documentclass[10pt,a4paper]{article}
\usepackage{graphicx}
\usepackage{amssymb,amsmath,amsthm,mathrsfs}
\usepackage{color}
\usepackage{eucal}
\numberwithin{equation}{section}
\newcommand{\ie}{\emph{i.e.}}

\newcommand{\cf}{\emph{cf}}
\newcommand{\e}{\mathrm{e}}

\newcommand{\R}{\mathbb{R}}
\newcommand{\Z}{\mathbb{Z}}

\newtheorem{claim}{Claim}[section]
\newtheorem{theorem}[claim]{Theorem}

\theoremstyle{definition}
\newtheorem{remark}[claim]{Remark}
\newtheorem{remarks}[claim]{Remarks}

\usepackage[normalem]{ulem}
\definecolor{DarkGreen}{rgb}{0,0.5,0.1}

\newcommand\soutD{\bgroup\markoverwith
{\textcolor{DarkGreen}{\rule[.5ex]{2pt}{1pt}}}\ULon}
\newcommand\soutB{\bgroup\markoverwith
{\textcolor{blue}{\rule[.5ex]{2pt}{1pt}}}\ULon}
\newcommand{\Hm}[1]{\leavevmode{\marginpar{\tiny%
$\hbox to 0mm{\hspace*{-0.5mm}$\leftarrow$\hss}%
\vcenter{\vrule depth 0.1mm height 0.1mm width \the\marginparwidth}%
\hbox to 0mm{\hss$\rightarrow$\hspace*{-0.5mm}}$\\\relax\raggedright
#1}}}

\begin{document}

\title{Spectral asymptotics induced by approaching and diverging planar circles }

\author{Sylwia Kondej}

\date{
\small \emph{
\begin{quote}
\begin{itemize}
Institute of Physics, University of Zielona G\'ora, ul.\ Szafrana
4a, 65246 Zielona G\'ora, Poland; s.kondej@if.uz.zgora.pl
\end{itemize}
\end{quote}
}
\medskip
\date{
\small \emph{
\begin{quote}
\begin{itemize}
Institute of Physics, University of Zielona G\'ora, ul.\ Szafrana
4a, 65246 Zielona G\'ora, Poland; s.kondej@if.uz.zgora.pl
\end{itemize}
\end{quote}
}
\medskip
29 March  2016}
%
}
\maketitle

\begin{center}
\emph{Dedicated to Pavel Exner on the occasion of his 70$^{th}$
birthday.}
\end{center}

%
\begin{abstract}
\noindent We consider two dimensional system governed by the Hamiltonian
with  delta
interaction supported by two concentric circles
separated by distance $d$. We analyze the asymptotics of
the discrete eigenvalues  for $d \to 0$ as well as for $d\to \infty$.
\end{abstract}
%



\section{Introduction}
The paper belongs to research often called Schr\"odinger operators with delta potentials.
We study a special model: two dimensional quantum system with delta potential supported by two concentric
circles: $C_{R}$ and $C_{R_d}$, where $C_R: =\{(x,y)\in \R^2\,:\, (x^2+y^2)^{1/2} =R \}$ and $C_{R_d} $
is defined analogously for
$R_d := R+d$, $d>0$. The Hamiltonian of such system can be symbolically written as
\begin{equation}\label{eq-delta}
  -\Delta -\beta  \delta_{C_{R}}- \alpha  \delta_{C_{R_d}}\,,\quad \mathrm{where} \quad  \alpha ,\,\beta \in \R\,,
\end{equation}
where   $\delta_{C_r}$ stands for the Dirac delta
supported on  $C_r$. To define a self-adjoint operator $H_{\alpha, \beta, d}$ coresponding formally to (\ref{eq-delta}) we employ the form sum method.
\\
\emph{The main results.} We investigate the behaviour of the discrete eigenvalues for $d\to 0$ and $d\to \infty $. In fact, in both asymptotics one can observe certain  "spectral memory" on  single circle system.
Therefore, it is convenient to introduce a special notation $H_{\gamma , R}$ for the Hamiltonian corresponding to
the formal expression
\begin{equation}\label{eq-delta1}
  -\Delta -\gamma   \delta_{C_{R}}\,,\quad \gamma  \in \R\,.
\end{equation}
In the following $\gamma $ will be expressed by means of the coupling constants $\alpha $ and $\beta $, however,
this dependence will be different in two considered cases.
If $\gamma >0$ then operator $H_{\gamma , R}$ has $2M_{\gamma, R} +1$ eigenvalues (counting multiplicity), where $M_{\gamma, R} :=\max \{m\in \Z \,:\, 2|m|<R \gamma  \}$.

The first result concerns the eigenvalue asymptotics in the
approaching circles system and the statement can be formulated as follows.

\medskip

$\bullet$ Let $E_m $ denote an eigenvalue of $H_{\alpha +\beta , R }$. Then the  eigenvalues of $H_{\alpha ,\beta , d}$ admit the following asymptotics
$$
E_m +t_m d +o(d)\,,
$$
for $d\to 0$. The explicit form for the first correction term $t_m$ is derived  in Theorem~\ref{th-1}. The analysed  system enables separation of variables and, consequently, relying on the implicit function theorem we  can reproduce $t_m$ in the terms of the Bessel functions and their derivatives. In~Section~\ref{sec-disc} we also study certain  properties of $t_m$; for example,
 we show that the sign of $t_m$ is  not defined generally.

\medskip

The second result is addressed to the system with circles separated by a large distance.

\medskip

$\bullet$ Assume that $d\to \infty $. Then the system has
a "tendency for the decoupling".  This is manifested  as the localization of eigenvalues of $H_{\alpha, \beta, d}$ near the eigenvalues of $H_{\beta , R}$ as well as $H_{\alpha, R_d}$. Precisely, the eigenvalues of $H_{\alpha, \beta, d}$ behave as
\begin{equation}\label{eq-asymdiv1} 
 \begin{cases}
    -\frac{\alpha^2}{4} + \frac{m^2- \frac{1}{4}}{d^2}\,  +o(d^{-2}  ) \,, & |m|\leq M_{\alpha, R_d} \,,  \\
     E_{m, \beta}   +w_m   \,\varepsilon +o(\varepsilon  ) \,, & |m|\leq M_{\beta , R}   \,,
  \end{cases}
\end{equation}
   where
   $
   \varepsilon := \exp (-2 d \kappa_{m,\beta } )
   $ and  $E_{m, \beta }$ stand for the eigenvalues of $H_{\beta , R}$. Note that the expression in the
   first line of
   (\ref{eq-asymdiv1}) reflects the asymptotics of eigenvalues of $H_{\alpha, R_d}$.
\\ \\
The models of delta interactions supported by circles or spheres has been already studied in the various contexts and dimensions,  see for example  \cite{AGS}, \cite{BT},  \cite{E}, \cite{EF}, \cite{EF08}, \cite{ET}, \cite{K12}, \cite{Shabani}.


\section{Preliminaries and the main  result} \label{sec-preliminaries}

\noindent \emph{Single ring: spectral properties of the system.}
 Spectral properties of the single circle Hamiltonian will
be essential for the  both asymptotics considered in this paper. Therefore, we start  our analysis from  recalling  some useful known facts, \cf~\cite{ET}.
Consider the Hamiltonian $H_{\gamma, R}$ associated to the sesquilinear form
$$
h_{\gamma, R} (f,g) = (\nabla f, \nabla g )_{L^2 (\R^2)} - \gamma \int_{C_R}\bar{f}g\,\mathrm{d}s  \,,\quad f, g \in W^{1,2}(\R^2)\,, \quad \gamma \in \R\,;
$$
where the functions in the second component are understood in the sense of the trace embedding $W^{1,2}(\R^2)\hookrightarrow L^2 (C_R)$ and the arc length parameter $s$ ranges $s\in [0, 2\pi R)$. In fact, $L^2 (C_R) $ can be identify with $L^2 ((0, 2\pi R))$.
Applying the results of \cite{BEKS} we conclude that the operator  $H_{\gamma , R}$ associated to $h_{\gamma, R}$ via the first representation theorem is self-adjoint. $H_{\gamma, R}$  gives a mathematical meaning to the formal expression (\ref{eq-delta}).

To be specific we introduce the polar system of coordinates $(r,\phi  )$ where  $r >0 $, $\phi \in [0, 2\pi )$
\footnote{In view of the above trace embedding operator it seems natural to consider $s$ parameter instead of $\phi $; however, since we going to implement the second circle it is more convenient to stay with standard polar coordinates.}.
The delta potential support $C_R$ decomposes $\R^2$ onto two disjoint open sets $\Omega^i$, $\Omega^e$; denote by $\bar\Omega^i$, $\bar \Omega^e$ their closures and $\mathcal{C}^1_R := C^1 (\bar\Omega^i  )\cup  C^1 (\bar \Omega^e )$. Assume that  $f\in \mathcal{C}^1_R  $ satisfies
\begin{equation}\label{eq-bcC1}
\begin{aligned} &
\lim_{r\to R^+} f (r,\phi  )=  \lim_{r\to R^-} f (r, \phi ) =:  f_R  (\phi ) 
\,, \\ &
\lim_{r\to R^+} \partial_r f (r,\phi )- \lim_{r\to R^-}\partial_r f (r,\phi )= -\gamma \,f_R (\phi )\,.
\end{aligned}
\end{equation}
Then the operator which acts as
\begin{equation}\label{eq-Ham1}
\check{ H}_{\gamma, R} f = -\Delta f
  \qquad \mbox{a.e.\ in } \R^2
  \,,
  \end{equation}
  on the domain
\begin{equation}
   D(\check{ H}_{\gamma, R})= \left\{
  f \in \mathcal{C}^1_R \cap W^{2,2}(\R^2 \setminus C_R)
  : f \mbox{ satisfies~\eqref{eq-bcC1}}
  \right\}
  \,
\end{equation}
is essentially self-adjoint  and its closure coincides with $H_{\gamma , R}$, \cf~\cite{BEKS}.
%

Since the delta potential is compactly supported the essential spectrum of $H_{\gamma , R}$
is stable under such ``perturbant``
, \ie
$$
\sigma_{\mathrm{ess}} (H_{\gamma , R}) = [0, \infty )\,,
$$
\cf~\cite{BEKS}.

Henceforth we will be interested in negative eigenvalues.
In view of the rotational symmetry 
 we  postulate that the  eigenfunctions
of $H_{\gamma, R}$ take the form
 $ \frac{1}{\sqrt{2\pi }}\varrho_m
(r)\e^{im \phi  }$, where $m\in \Z$.  Let $\kappa >0$.
The behaviour of eigenfunctions  at the infinity and origin imposes
\begin{align} \label{eq-bcC3} && \varrho _m
(r) = c_1 K_m (\kappa r) \,,\,\,\,\, \mathrm{for} \,\,\,\, r>R \,, \\
\nonumber && \varrho _m (r) = c_2 I_m (\kappa r)\,,\,\,\,\,
\mathrm{for} \,\,\,\, r<R\,,
\end{align}
where
$K_m (\cdot )$ and $I_m (\cdot)$ denote the modified Bessel functions,
\cf~\cite{AS}.
Using the boundary conditions (\ref{eq-bcC1}) we get the following spectral condition
\begin{equation}\label{eq-speq1}
K_{m}(\kappa R) I_{m}(\kappa R) = \frac{1}{\gamma R}\,,\quad m\in \Z\,,
\end{equation}
\cf~\cite{ET}.

It follows from (\ref{eq-bcC3}) that $\kappa $ determines the spectral parameter
and the solutions of (\ref{eq-speq1}) reproduce negative
eigenvalues $E$ of $H_{\gamma , R}$ by means of the relation $E=-\kappa^2$.

\begin{remarks} \label{rem-1}
A. Relying on asymptotics formulae (\ref{eq-asyminftyI})-(\ref{eq-asym0K}) and using the fact that $(K_m I_m )(\cdot )$ is monotonously decreasing we state
that the equation (\ref{eq-speq1}) has exactly one solution for $\kappa >0$ provided $2|m|<R\gamma $ or equivalently~$|m| \leq M_{\gamma , R}$ and no solution otherwise; recall that the notation $M_{\gamma, R}$ was introduced in introduction.
\\
B. It is also useful  to recall     that for $\gamma R  $ large the solution of (\ref{eq-speq1})
behaves as
$$
-\kappa_m^2 = -\frac{\gamma ^2}{4}+\frac{m^2 - \frac{1}{4}}{R^2} +O(\gamma ^{-2}R^{-4})\,.
$$
\end{remarks}

\emph{Hamiltonian with the delta potential supported by
two concentric rings.}  Let $s\in [0, 2\pi R)$ and $s_d\in [0, 2\pi R_d )$
stand for the arc length parameters associated to $C_R$ and $C_{R_d}$ respectively. For  $\alpha ,\beta \in \R$ let us define the following sesquilinear form
$$
h_{\alpha, \beta ,d} (f,g) = (\nabla f, \nabla g )_{L^2 (\R^2)} - \beta \int_{C_R}\bar{f}g\,\mathrm{d}s - \alpha \int_{C_{R_d}}
\bar{f}g
\mathrm{d}s_d \,,\quad f, g \in W^{1,2}(\R^2)\,,
$$
where we employ the trace  embedding of $W^{1,2}(\R^2)$ to $L^2 (C_r) \simeq L^2 ((0, 2\pi r))$, $r=R, R_d$.
Similarly as for the single circle case we define the operator  $H_{\alpha, \beta ,d}$ associated to $h_{\alpha, \beta ,d }$ via the first representation theorem.

Analogously for the single circle we can characterize $H_{\alpha, \beta, d}$ by means of boundary conditions.
Note that circles $C_R$ and $C_{R_d}$ decompose   $\R^2$ onto three  open sets $\Omega^i_1$,  $\Omega^i_2$  and  $\Omega^e$. Denote  $\mathcal{C}^1_{R, R_d} := C^1 (\bar\Omega^i_1 )\cup C^1 (\bar\Omega^i_2 )\cup C^1 (\bar \Omega^e )$ and assume that  $f\in \mathcal{C}^1_{R, R_d} $ satisfies
\begin{equation}\label{eq-bcC2}
\begin{aligned}
&\lim_{r\to R^+} f(r,\phi ) =  \lim_{r\to R^-} f (r,\phi  ) =: f_R (\phi )\,, \\
& \lim_{r\to R^+}\partial_r f (r,\phi ) - \lim_{r\to R^-}\partial_r f (r,\phi  )= -\beta f_R (\phi  )\,,\\
& \lim_{r\to R_d^ +} f (r,\phi ) =  \lim_{r\to R_d^-} f ( r,\phi  ) =: f_{R_d} (\phi )\,,\\ 
& \lim_{r\to R_d^+}\partial_r f (r,\phi )- \lim_{r\to R_d^-}\partial_r f (r,\phi )= -\alpha   \,f_{R_d}
(\phi )\,.
\end{aligned}
\end{equation}
In fact, $H_{\alpha, \beta, d }$ stands for the closure
\begin{equation}\label{eq-Ham1}
\begin{aligned}
  \check{H}_{\alpha, \beta, d } f  &= -\Delta f
  \qquad \mbox{a.e.\ in } \R^2
  \,,
  \\
  D( \check{H}_{\alpha, \beta, d } ) &= \left\{
  f \in \mathcal{C}^1_{R, R_d}\cap  W^{2,2}(\R^2 \setminus ( C_R \cup C_{R_d}) \big)
  : f \mbox{ satisfies~\eqref{eq-bcC2}}
  \right\}
  \,.
\end{aligned}
\end{equation}

\subsection{Spectral equation for the double ring system}
To derive  the spectral equation for the double ring system we proceed analogously as in the previous
case.
 The system again admits separation variables.
Consequently,  the eigenfunctions of $H_{\alpha, \beta, d}$ can be written
as  $ \frac{1}{\sqrt{2\pi }}\rho _m (r)\e^{im\phi }$,  where $m\in \Z$ and
\begin{eqnarray} \nonumber
&&\rho _m (r) = C_1 K_m (\kappa r) \,,\,\,\,\, \mathrm{for}
\,\,\,\, r>R_d \,,\\ \nonumber && \rho _m (r) = C_2 K_m
(\kappa r) +C_3 I_m (\kappa r)\,,\,\,\,\,\mathrm{for} \,\,\,\,
R<r<R_d
\end{eqnarray}
and
$$\rho _m (r) = C_4 I_m (\kappa r) \,,\,\,\,\, \mathrm{for} \,\,\,\,
r<R \,;
$$
\\
Inserting the above formulae to (\ref{eq-bcC2})  we obtain four equations
\begin{eqnarray} \nonumber
&& C_1 K_m (\kappa R_d  ) - C_2 K_m (\kappa R_d  ) -C_3 I_m
(\kappa R_d   ) =0\,,\\ \nonumber && C_1 (\kappa K'_m (\kappa R_d
)+\alpha K_m (\kappa R_d   )) -C_2
\kappa K'_m (\kappa R_d  )-C_3 \kappa I'_m (\kappa R_d   ) =0\,,\\
\nonumber && C_2 K_m (\kappa R  )  +C_3 I_m (\kappa R  ) - C_4 I_m
(\kappa R  )=0 \,,\\ \nonumber && C_2 \kappa K'_m (\kappa R  )+
C_3 \kappa I'_m (\kappa R  ) +\ C_2(\beta  I_m (\kappa R  ) -
\kappa I'_m (\kappa R  ))=0\,.
\end{eqnarray}
\emph{Spectral equation.} The above system of equations admits a
solution $iff$ the determinant of the corresponding
 matrix vanishes. This condition can be written by means of
  the following equation
\begin{equation}\label{eq-speqC}
  \eta _m (\kappa , d) = 0\,,\quad m\in \Z\,,\,\kappa >0\,,\,d\geq 0\,,
\end{equation}
where
\begin{eqnarray}\label{eq-defeta}
  \eta_m  (\kappa , d) =   \nu_{m}(\kappa , d )
- \xi_{m, \alpha } (\kappa ) \xi_{m, \beta}  (\kappa ) \,,
\end{eqnarray}
with  \begin{equation}\label{eq-xi}
\xi_{m, \alpha } (\kappa ) \equiv \xi_{m, \alpha , d } (\kappa ):= \alpha R_d (K_m I_m )(\kappa R_d ) - 1\,, \quad
\xi_{m, \beta}  (\kappa ):= \beta R (K_m I_m )(\kappa R )-1\,,
\end{equation}
and
\begin{equation}\label{eq-nu}
  \nu_{m}(\kappa , d ):=  \alpha \beta R_d R\, K_m^2(\kappa R_d)   I_m ^2 (\kappa R  )
  \,.
\end{equation}
Formulae~(\ref{eq-speqC})  constitute the spectral equations for
$H_{\alpha, \beta , d}$.

\begin{remark} Note that the functions $\xi_{m, \tau} $, where $\tau =\alpha, \beta $ are related to the single circle systems. More precisely, the relations
 $$
\xi_{m, \tau}(\kappa ) = 0 \,,
$$
determine spectral equations of $H_{\alpha , R_d }$ and $H_{\beta, R}$.
\end{remark}

\section{Approaching  rings }

In this section we consider the eigenvalue asymptotics for $d \to 0$ and $\alpha +\beta >0$.
Note that for $d =0 $ equation (\ref{eq-speqC}) reads
\begin{equation}\label{eq-speqC0}
K_m(\kappa R )
  I_m (\kappa R  ) =\frac{1}{(\alpha +\beta )R}\,;
\end{equation}
the latter  corresponds to the single ring Hamiltonian $H_{\gamma , R}$ with coupling constant
$\gamma =\alpha +\beta $, \cf~(\ref{eq-speq1}).   As follows from the previous discussion, see Remark~\ref{rem-1},  equation (\ref{eq-speqC0}) has exactly one solution $\kappa_m$ provided $|m| \leq  M_{\alpha +\beta, R}$.

The following theorem provides  the spectral asymptotics for approaching rings.
\begin{theorem} \label{th-1}
Assume $\gamma= \alpha +\beta >0$.
Let $E_m$, where   $|m| \leq  M_{\gamma , R}$, stand for an  eigenvalue of $H_{\gamma, R}$.
Then the eigenvalues of $H_{\alpha, \beta, d}$ admit the following asymptotics for $d\to 0$:
\begin{equation}\label{eq-asymC}
  E_{m} (d)= E_m  +t_m\,d
  +o(d)\,,
\end{equation}
   where
  $t_m$  is given by
  \begin{equation}\label{eq-defS}
  t_m := \frac{2 \kappa_m I_m K_m  \left(-\alpha \beta R  I_m K_m   +\alpha \kappa_m
R (I_m K_m )'  +\alpha
 I_m K_m \right)}{R(I_m K_m )'}\,;
\end{equation}
moreover, functions $K_m(\cdot )$ and $I_m(\cdot )$  as well as their derivatives
contributing to   (\ref{eq-defS}) are defined for the value $R\kappa_m$.
\end{theorem}
\begin{proof}
Suppose $|m|\leq M_{\gamma, R}$.
 Eigenvalues of $H_{\alpha, \beta, d}$ are determined by the solutions of  (\ref{eq-speqC}).  Note that
$$
\eta_m (\kappa_m , 0) = 0\,.
$$
Using the regularity of $K_m$ and $I_m$ we state for $d \in \R$ and $\kappa >0$ the functions
$ \frac{\partial \eta_m
}{\partial \kappa }  $ and $ \frac{\partial \eta_m
}{\partial d }  $ are  $C^\infty $. Furthermore,
using (\ref{eq-speqC0}) we
get
\begin{equation}\label{eq-der2}
\frac{\partial \eta_m  }{\partial \kappa } = (\alpha +\beta ) R^2
(I_m K_m )' = R \frac{(I_m K_m  )'}{(I_m K_m )} \,,
\end{equation}
where the derivative at the left hand side is defined at $(\kappa_m, 0)$. Moreover,  $Z_m  =Z_m (R\kappa_m)$, $Z_m= K_m \,, I_m$ and the analogous notation is
applied for the derivatives contributing to right hand side of
the above
expression. Since the function $(I_m K_m) (\cdot )$ is monotonously decreasing we
have $\frac{\partial \eta_m  }{\partial \kappa } <0$.
Consequently, we can employ the implicit function theorem which states that there exists a neighbourhood
$U\in \R$ of $0$ and the unique function $U \ni d\mapsto \kappa_m (d) \in \R$ such that $\eta_m (\kappa_{m}(d ) , d)=0$ and
\begin{equation}\label{eq-expkappa}
\kappa _{m} (d)= \kappa_m  -\left(\frac{\partial
\eta_m }{\partial d }\right) \left(  \frac{\partial \eta_m
}{\partial \kappa } \right)^{-1} \, d +o(d)\,,
\end{equation}
where all derivatives in the second component are determined for
$d =0$, $\kappa =\kappa_m$. \\
Using (\ref{eq-defeta}) and the Wroskian equation
\begin{equation}\label{eq-wron}
(I'_mK_m  )(z) - (K_m'I_m)(z)=\frac{1}{z}\,
\end{equation}
we get by a straightforward calculation
\begin{eqnarray} &&
\frac{\partial \eta_m }{\partial d }=
\\ \nonumber && -\alpha \beta R (K_m I_m) +\alpha
\kappa_m R  (I_m K_m  )' +\alpha (I_m K_m ) \,.
\end{eqnarray} 
Combining  the above derivatives together with (\ref{eq-expkappa}) we arrive at
\begin{eqnarray}\label{eq-expev}
&&  E_{m} (d)  =  -\kappa_{m} (d) ^2 = \\
\nonumber &&
 -\kappa_m ^2 +2 \kappa_m  \frac{\partial \eta_m
}{\partial d}\left(  \frac{\partial \eta_m }{\partial
\kappa } \right)^{-1} d +o(d ) =  E_m  + t_m
d +o(d)\,,
\end{eqnarray}
with $t_m$  given by (\ref{eq-defS}).
\end{proof}

\subsection{ Discussion on the first order correction} \label{sec-disc}

In this section we discuss some properties of the first order correction  for converging rings.\\

\emph{The first order correction in the terms of unperturbed eigenfunctions.} The following analysis will be conducted for $m=0$.  Let $f_0$ stand for the ground state of $H_{\alpha+ \beta, R }$. Using  (\ref{eq-bcC1})
and (\ref{eq-bcC3}) we conclude  that $c_1 = \frac{I_0 }{K_0 }c_2$; recall $K_0  =K_0
(R\kappa_0 )$, $I_0 = I_0
(R\kappa_0 )$ and the analogous notation is employed   for derivatives.   Applying  the relation
$$
\int x Z_0^2 (x)\mathrm{d}x= \frac{x^2}{2} (Z_0 ^2 (x)-(Z_0' (x)) ^2\,,\quad Z_0=I_0\,, K_0
\,,
$$
one can show that  the norm of eigenfunction $ f_0$ is
given by
$$
\| f_0 \|^2_{L^2 (\R^2)} = | c_2|^2 \frac{R^2}{2 K_0^2} ((IK')^2 - (K I')^2)\,.
$$
Using again the Wronskian equation (\ref{eq-wron}) one gets
$$
\| f_0 \|^2_{L^2 (\R^2)} =-| c_2|^2\frac{R}{2\kappa_0  K_0 ^2 }
(I_0K_0)'\,.
$$
Applying  the above formula   
together with  boundary conditions (\ref{eq-bcC1}) and
comparing this with (\ref{eq-defS}) we obtain
\begin{equation}\label{eq-defS2}
t_0=\frac{-\alpha (\int_{C_R}\partial ^+_r|f_0 |^2\mathrm{d}s
+\alpha  \int_{C_R}|f_0 |^2 \mathrm{d}s
)-\frac{\alpha}{R}\int_{C_R}|f_0 |^2\mathrm{d}s }{\|f_0
\|_{L^2 (\R^2)}^2}\,,
\end{equation}
where we abbreviate
$\partial ^+_r f(r,\phi) = \lim_{r\to R^+}\partial_r f(r,\phi) $; recall that the equation $s=R\phi $ states the relation
the relation between $\phi \in [0,2\pi)$ and $s \in [0,2\pi R)$.

Let us mention that  the above formula  describes a very particular case of the class considered in the  forthcoming paper~\cite{KK16}. In this paper the  spectral asymptotics for approaching hypersurfaces in $\R^d$ is analyzed.
The method developed  in \cite{KK16} allows to reconstruct asymptotics of eigenvalues by means of the ''unperturbed`` eigenfunctions. The technics enables generalization for
complex coupling constants.

The last component of (\ref{eq-defS2}) reflects   contribution of the curvature to the first correction term.
More general situation shows  a presence of the first  mean curvature in eigenvalue asymptotics,~\cf~\cite{KK16}.
Furthermore, let us note that a contribution of the first mean curvature
in spectral asymptotics has been recently showed
in related problems,
see~\cite{Kr09}, \cite{K10} and~\cite{Pankrashkin-Popoff_2015b}.

The second component of (\ref{eq-defS2}) is a consequence of singular character of delta
potential. Suppose $f_0$ and $f_d$ denote the normalized ground states of $H_{\alpha +\beta , R}$ and $H_{\alpha, \beta , d}$, respectively. The second component of (\ref{eq-defS2}) comes directly from the fact that $\partial_r (f_0  - f_d) (r,\phi )  $ do not tend to $0$
if $d\to 0$ and $r\in (R, R_d)$.
\\ \\
 \emph{The sign of $t_m$.}  For the one dimensional system with two  converging points of interaction
the first order correction is always \emph{positive}, \cf~\cite{AGHH}. This means
that the splitting of the singular potential from one point to two points
leads to pushing up the
eigenvalue.
The situation is slightly different in the case of converging circles. As formula (\ref{eq-defS}) shows  the
sign of $t_0$ depends on
$$
\varsigma := \alpha (1 - \beta R )I_0K_0 +\alpha \kappa_0 R(I_0 K_0)'\,,
$$
\ie~$\mathrm{sign} \,\varsigma  = -\mathrm{sign}\, t_0 $.

$\bullet$ First, let us consider the situation when $R\to 0 $; then $\kappa_0 \to 0$ as well, \cf.~\cite{ET}. Employing
asymptotics formulae for $Z_0$, where $Z_0 =I_0 , K_0$, see (\ref{eq-asym0I}) and (\ref{eq-asym0K})
together with (\ref{eq-recur1})-(\ref{eq-Ki})  one gets
$$
\varsigma \sim - \alpha (1-\beta R ) \ln (\kappa_0 R)\,,
$$
which implies $\varsigma >0$  for $R$ small enough and, consequently, leads to  $t_0 <0$.

$\bullet$ Now we assume that $R\to \infty $. Then $\kappa_0 \sim \alpha /2 $ and using again
formulae (\ref{eq-asym0I}) and (\ref{eq-asym0K})    we obtain:
\begin{equation}\label{eq-}
  \varsigma =  - \alpha \beta R \left(\frac{1}{2\kappa_0 R}+ O\left((\kappa_0 R)^{-3} \right)\right)\,.
\end{equation}
This shows that for $R$ large enough we have $\varsigma <0$ and $t_0 >0$.
\\
The above discussion establishes that  the sign of the first order correction
term is generally undefined.

\section{Diverging rings}

In this section we consider the asymptotics for circles separated by a large distance, \ie~for $d\to \infty$.


Following the convention introduced in the previous discussion we denote by $H_{\alpha, R_d}$~and~$H_{\beta , R}$ the corresponding  single circle Hamiltonians.  Operator $H_{\alpha , R_d}$ has $2M_{\alpha, R_d} +1$ (counting multiplicities) eigenvalues $\{E_{m, \alpha }\}_{|m|\leq M_{\alpha, R_d} }$ and
$H_{\beta, R }$ has $2M_{\beta, R} +1$ eigenvalues $\{E_{m, \beta }\}_{|m|\leq M_{\beta, R} }$. Suppose $\tau = \alpha, \beta $. In fact, $E_{m, \tau}$  can be recovered from the spectral equations, \ie~$E_{m, \tau}= - \kappa_{m, \tau}^2 $ where  $\kappa_{m,\tau }$ stand for the solutions of
$$
\xi_{m, \tau  }(\kappa )= 0\,,\quad 
\mbox{  for   }\kappa >0\,.
$$
Recall that   $\xi_{m,\tau }$ are defined by (\ref{eq-xi}).
Moreover, using the statement of Remark~\ref{rem-1} we conclude that
 \begin{equation}\label{eq-asympinftyev}
 E_{m, \alpha }=   -\frac{\alpha^2}{4} + \frac{m^2- \frac{1}{4}}{R_d^2}\,  +O(d^{-4}  )\,.
 \end{equation}

\begin{theorem} \label{th-2} Assume that $\alpha $ and $\beta $ are positive and $E_{m, \beta } \neq -\frac{\alpha^2}{4}$ for
all $|m|\leq M_{\beta , R}$. Then the eigenvalues  of
$H_{\alpha, \beta ,d}$ admit the following asymptotics for $d\to \infty$:
\begin{equation}\label{eq-efHepsilon}
 \epsilon_d =  \begin{cases}
    -\frac{\alpha^2}{4} + \frac{m^2- \frac{1}{4}}{d^2}\,  +o(d^{-2}  ) \,, & |m|\leq M_{\alpha, R_d } \,,  \\
    E_{m, \beta}   +w_m   \,\varepsilon +o(\varepsilon  ) \,, & |m|\leq M_{\beta, R}   \,,
  \end{cases}
\end{equation}
   where
   \begin{equation}\label{eq-w}
   \varepsilon := \exp (-2 d \kappa_{m,\beta } )\,,\quad \quad 
  w_m :=   \frac{\pi \alpha \beta R\e^{-2\kappa_{m,\beta } R} I_m (R\kappa _{m, \beta })^2 }{ (1  -\alpha /(2\kappa _{m, \beta }))\xi ' _{m,\beta }(\kappa _{m ,\beta} ) }\,.
   \end{equation}
   \end{theorem}

\begin{remark}   In fact, $\epsilon_d $   reflects the asymptotics of $2(M_\alpha +M_\beta )+2$ eigenvalues of $H_d$. However, since  $\epsilon_d $  converge to $E_{m, \alpha  }$ and $E_{m,\beta}$ we leave the labelling inherited from the discrete eigenvalues of the single circle Hamiltonians.
\end{remark}

\begin{proof} The analysis is based on investigating spectral equation  (\ref{eq-speqC}) which reads
\begin{equation}\label{eq-spectmain}
  \eta_m (\kappa , d ) = \nu_m (\kappa , d) - \xi_m (\kappa , d) =0 \,,
\end{equation}
where
$$
 \xi_m (\kappa , d):=   \xi_{m, \alpha , d  } (\kappa ) \xi_{m, \beta} (\kappa )\,.
$$
First, assume that $|m|\leq M_{\beta , R}$. Then for $d$ large enough we have
$|m|\leq M_{\alpha , R_d}$. Combining equations (\ref{eq-xi}) and (\ref{eq-nu}) with the  formulae
(\ref{eq-asyminftyI}) and (\ref{eq-asyminftyK}) we get the following asymptotics for $\kappa \to \infty $ and any $m\in \Z$:
\begin{equation}\label{eq-asymnuxi}
  \begin{cases}
    \xi_{m} (\kappa , d ) = 1 -\frac{\alpha+\beta }{2\kappa } +O (\kappa^{-2})\,, \\
    \nu _{m} (\kappa , d ) = \frac{\alpha \beta }{4\kappa^2} \mathrm{e }^{-2d\kappa} \big(1+o_\kappa (1) \big)    \,; \\
  \end{cases}
\end{equation}
the  error terms in the above expressions are uniform with respect to  $d>C$ where $C$ is a positive
number \footnote{Note that in this proof $C$ denotes a positive
constant which can change from line to line}. The symbol $o_\kappa (1)$ donotes that the asymptotics understood
with respect to $\kappa $. On the other hand, for $\kappa \to 0 $ we have
\begin{equation}\label{eq-asymnuxi-0}
  \xi_{m} (\kappa , d ) = \begin{cases}
    \big( \frac{\alpha R_d }{2m } - 1\big)  \big( \frac{\beta R }{2m } - 1 \big)\big(1+o_\kappa (1) \big)
    \,, & m\neq 0  \\
     \alpha \beta  \log (\kappa R)  \log (\kappa R_d ) RR_d \big(1+o_\kappa (1) \big)
          \,, & m=0 \, \\
  \end{cases}
\end{equation}
and
\begin{equation}\label{eq-asymnunu-0}
 \nu_{m} (\kappa , d ) =  \begin{cases}
     \frac{\alpha }{4m^2} \frac{R^{2m+1}}{R_d^{2m+1}}\big(1+o_\kappa (1) \big)
    \,, & m\neq 0  \\
    \alpha \beta   \log^2 (\kappa R_d) RR_d \big(1+o_\kappa (1) \big)
          \,, & m=0 \,, \\
  \end{cases}
\end{equation}
where all error terms are uniform with respect to $d>C$. We have $\nu_m (\kappa, d ) >0$. Moreover, $\nu_m (\kappa, d ) \to 0 $ as $d\to \infty $ and the limit
is uniform with respect to $\kappa >C$.
It follows from (\ref{eq-asymnuxi-0}) and (\ref{eq-asymnunu-0}) that if $m\neq 0$ then $\xi_m (0, d )> \nu_m (0, d )$ for $d$ large enough. If $m=0$ then  the corresponding limits for $\kappa \to 0$ do not exist, however,
$\xi_m (\kappa , d )> \nu_m (\kappa , d )$ holds for $\kappa $ from a neighbourhood of $0$ and $d$ large enough. The function $\xi_m (\kappa, d)$ has two roots:
$\kappa_{m, \alpha } $ and $\kappa_{m, \beta }$. Moreover  $\xi_m (\kappa, d ) \to 1 $  for $\kappa \to \infty $
and the limit is uniform with respect to
$d >C$. It follows from the discussed properties of $\xi_m$ and $\nu_m$ and their continuity that equation (\ref{eq-spectmain}) has at least two solutions. Furthermore,
for $d\to \infty $ we have $\xi_m (\kappa, d ) = (\alpha /(2\kappa) -1 +o_d (1)) \xi_{m, \alpha }(\kappa )$
and the error is uniform with respect to $\kappa >C$. This implies that (\ref{eq-spectmain}) has exactly two
solutions for $\kappa >C$ which we denote as $\kappa_{m, \alpha }(d)$ and $\kappa_{m, \beta}(d)$ and  they
 approach to $\kappa_{m, \alpha }$ and $\kappa_{m, \beta  }$ as $d\to \infty$.  Note that, since $C$ can be chosen arbitrary small
we can conclude, in view of the behaviour of $\nu_m$ and $\xi_m$
in a neighbourhood of $0$, that (\ref{eq-spectmain}) has  exactly two solutions for $\kappa \in (0, \infty )$. Let us
 consider $\kappa_{m, \beta  }(d)$. We have
$$
\kappa_{m, \beta  }(d)= \kappa_{m, \beta  } + \delta_{m, \beta  } (d)\,,
$$
where $\delta_{m, \beta} (d)$ converges to $0$ as $d\to \infty $.  Inserting $\kappa_{m, \beta  }(d)$ to $\eta (\kappa , d)$ one gets
\begin{equation}\label{eq-eta1}
\eta (\kappa_{m, \beta  }(d), d)= \nu_m (\kappa_{m, \beta } (d), d  ) - \xi_{m, \alpha } (\kappa_{m, \beta  })  \xi'_{m, \beta  } (\kappa_{m, \beta  } ) \delta_{m, \beta }+o(\delta_{m, \beta })\,.
\end{equation}
The equation
\begin{equation}\label{eq-eta}
\eta (\kappa_{m, \beta  }(d), d)=0
\end{equation}
 and the behaviour of $\nu_m (\kappa_{m, \beta } (d), d  )$ imposes
$d\delta _{m, \beta }(d)\to 0 $ as $d\to \infty $. Therefore, we get
$$
 \nu_m (\kappa_{m, \beta } (d), d  ) = \frac{\pi \alpha \beta R }{2 \kappa_{m, \beta } }  \e^{-2\kappa_{m,\beta } R} I_m (R\kappa _{m, \beta })^2  \varepsilon + o(\varepsilon)\,.
$$
Implementing the above expression and (\ref{eq-eta1}) to (\ref{eq-eta}) and comparing appropriated terms  leads to (\ref{eq-w}). To derive the second eigenvalue
we employ  (\ref{eq-asympinftyev}) together with the asymptotics of $\xi_{m, \alpha }$ which depends on $d$ as well.
The analogous analysis as above  establishes the  asympotics of eigenvalues localized near $-\alpha^2 /4$, see the first line of (\ref{eq-efHepsilon}).

If $|m|>M_{\beta, R}$ and $|m|\leq M_{\alpha, R_d}$ then $\xi_m (\kappa, d )=0$ has only one solution $\kappa_{m, \alpha }$.
Then repeating the above steps one shows the  existence of one solution of the spectral equation; this solution admits the asymptotics specified in the first line of (\ref{eq-efHepsilon}).
\end{proof}

The result of the above theorem corresponds to the phenomena known for regular potentials. It was shown in \cite{Harrell}
that the introducing a
second well to the  single well system leads to the splitting of original eigenvalues and the corresponding spectral gaps can be expressed by the  current. Consequently,
the asymptotics of the gaps, if the wells are  separated by a large distance,
is determined by the exponential decay of  eigenvectors. Theorem~\ref{th-2} shows that the introducing  interaction supported by circle $C_{R_d}$
to the system governed by $H_{\beta, R} $  leads to the shifting of  original energies $E_{m,\beta }$
and this spectral shifting is determined by exponential decay of corresponding eigenfunctions.

One the other hand, the system governed by $H_{\alpha , R_d }$ admits eigenvalues $E_{m, \alpha }$  which depend  on $d$, see (\ref{eq-asympinftyev}). Formula (\ref{eq-efHepsilon}) shows that the leading terms of this eigenvalues asymptotics are preserved
if we introduce  also interaction supported by $C_R$.

\section{Appendix}
We complete here the asymptotics of functions contributing to the spectral equations, see~\cite{AS}.
Namely, for $z \to \infty$ and $m\in \Z$ we have
\begin{equation}\label{eq-asyminftyI}
  I_m (z) = \frac{\e^z}{\sqrt{2\pi } } \left( 1-\frac{4m^2 -1}{8z }+O(z^{-2})\right)\,,
\end{equation}
and
\begin{equation}\label{eq-asyminftyK}
  K_m (z) = \sqrt{\frac{\pi }{2 z}}\, \e^{-z} \left( 1+ \frac{4m^2 -1}{8z }+O(z^{-2})\right)\,.
\end{equation}
Furthermore, for $z\to 0 $ we have
\begin{equation}\label{eq-asym0I}
  I_{m} (z) \sim \frac{1}{\Gamma (m+1)} \left(\frac{z}{2}\right)^m \,,\quad m\in \Z\,
\end{equation}
and
\begin{equation}\label{eq-asym0K}
  \begin{cases}
    K_{m} (z) \sim \frac{\Gamma (m)}{2} \left(\frac{2}{z}\right)^m \,,\quad  m\neq 0 \\
   K_{0} (z) \sim -\ln ( z/2)    \,.
  \end{cases}
\end{equation}
Recall $\Gamma (m)= (m-1)!$.

For the asymptotics of derivatives the following formulae will be useful
\begin{equation}\label{eq-recur1}
  I'_m (z) = \frac{I_{m-1}(z)+I_{m+1 }(z)}{2}
\end{equation}
and \begin{equation}\label{eq-recur2}
  K'_m (z) =  - \frac{K_{m-1} (z) +K_{m+1 } (z)}{2}\,.
\end{equation}
Furthermore, since $Z_m = Z_{-m}$, where $Z_m = I_m, K_m$ we, for example, have
\begin{equation}\label{eq-Ki}
   I'_0 (z) =  I_1 (z)\,,\quad \quad     K'_0 (z) = - K_1 (z)\,.
\end{equation}

\subsection*{Acknowledgements}
The author is grateful to the referee for reading carefully
the first version of the paper and for the useful remarks and  suggestions of changes. The work was  supported
by the project with the decision No. DEC-2013/11/B/ST1/03067 of the Polish National Science Centre.

%

\providecommand{\bysame}{\leavevmode\hbox
to3em{\hrulefill}\thinspace}
\providecommand{\MR}{\relax\ifhmode\unskip\space\fi MR }
\providecommand{\MRhref}[2]{%
  \href{http://www.ams.org/mathscinet-getitem?mr=#1}{#2}
} \providecommand{\href}[2]{#2}

\end{document}